\documentclass[conference]{IEEEtran}

\usepackage{amsmath}
\usepackage{amsfonts}
\usepackage{amsthm}
\usepackage{fixltx2e}
\usepackage{graphicx}
\usepackage{epstopdf}

\usepackage{algorithm}
\usepackage{algorithmicx}
\usepackage{algpseudocode}
\usepackage{xcolor}

\begin{document}

\theoremstyle{plain}
\newtheorem{lemma}{Lemma}
\newtheorem{theorem}{Theorem}
\newtheorem{prop}{Proposition}
\newtheorem{remark}{Remark}
\newtheorem{cor}{Corollary}
\newtheorem{subgroup}{Subgroup}
\newtheorem{definition}{Definition}

\title{Capacity and Energy-Efficiency of Delayed Access Scheme for Small Cell Networks}

\author{\IEEEauthorblockN{
 Haluk Celebi\IEEEauthorrefmark{1} and
 \.{I}smail G\"{u}ven\c{c}\IEEEauthorrefmark{2}
 and Henning Schulzrinne\IEEEauthorrefmark{3}
 }\\
\IEEEauthorblockA{\IEEEauthorrefmark{1}Department of Electrical Engineering, Columbia University, New York, NY\\
\IEEEauthorrefmark{2}Department of Electrical and Computer Engineering, North Carolina State University, Raleigh, NC\\
\IEEEauthorrefmark{3}Department of Computer Science, Columbia University, New York, NY\\
\IEEEauthorblockA{Email: {\tt haluk@ee.columbia.edu}, {\tt iguvenc@ncsu.edu}, {\tt hgs@cs.columbia.edu  }
}
}}

\maketitle

\begin{abstract}
Data applications may typically tolerate a moderate delay before packet transmission between user equipment (UE) and cell begins. This delay can be taken advantage to reduce the communication distance, improve coverage probability, and increase overall energy-efficiency of the small cell network. To demonstrate such merits, we suggest a simple access scheme and analyze the distribution of coverage probability and throughput as a function of delay and transmit distance. Sufficient
number of small base stations (SBSs) handle the peak traffic load. To improve energy-efficiency of the network, a number of SBSs are switched off at low traffic periods. Energy-efficiency can be further improved by turning all of the SBSs on and off, rather than selecting a subset and leaving them off. By doing so, coverage probability and bit-rate can be improved by delaying their transmissions and waiting for a closer SBS to become available. Results show that by turning SBSs on and off continuously and taking advantage of initial delay to connect a SBS yield an order of magnitude improvement in energy-efficiency, improves the coverage probability significantly at low signal to interference and noise (SINR) regime.

\end{abstract}

\begin{IEEEkeywords}
5G, delay tolerant network (DTN), delayed access, energy efficiency, latency, small cells, URLLC
\end{IEEEkeywords}

\IEEEpeerreviewmaketitle

\section{Introduction}

 Today, due to rising traffic demand, cellular networks are on the crossroads of a major shift in terms of their structure, topology, and operation. These major changes compel new areas of research for capacity enhancement and energy optimization. The key performance objective of 5G and beyond is to increase capacity by several orders of magnitude in terms of aggregate throughput \cite{5Gwhitepaper} in ten years. This objective will be realized  mainly by densely deploying small cells, utilizing new spectrum in micro and mili-meter wave (mmWave) bands, and adopting multi-antenna and multi-point transmission schemes. At this point the concern arises whether energy-per-bit will remain affordable with the deployment of small cells.  Basic solution to the energy concern is to turn off small cells at low traffic periods.

 Simple and self organizing access strategies are vital to cope with the increasing complexity and to realize highly flexible connectivity in 5G. Current cellular networks will be 50$\times$ densified with the deployment of small cells. The densification will increase the operational complexity about 60 times \cite{Magazine:5GChallengesSON}. Therefore, self-management, adaptability, and decentralized access strategies become inevitable. Besides, forecasts indicate that machine type devices such as sensors, actuator nodes, tablets, drones, and electronic meters will outnumber  smart phones. Each of these devices may have different service demand, and hence, 5G networks will have to brace for flexible and data-oriented access schemes. For many applications, data communications can tolerate a modest delay before data transmission begins, which helps designing flexible access schemes. By delaying users' access to the cell, it is possible to decrease communication distance even if majority of cells are turned off for saving energy. It then becomes of interest to study the associated trade-offs between the network energy efficiency and the access latency to the network, which to our best knowledge has not been explored in the existing literature.

There is a rich literature on the analysis of SINR distribution in random networks using stochastic geometry. In \cite{IEEEjor:TractableJeff}, tractable expressions of coverage probability and bit-rate of a user in a random location are given. In \cite{StochGeoOnOff}, control signals are muted to improve SINR and save power in an adaptive manner. In \cite{smallCellSleepStrategies}, bit-rate and energy efficiency (EE) analysis of random cellular networks are given. Then, EE maximization algorithms are offered. Besides,  obtaining the trade-off between access delay and bit-rate may be of importance for emerging ultra-reliable and low-latency communications (URLLC) \cite{URLLCBennisLetter,DBLP_FedLearnURLLV2V}.
The novelty of our work over such existing literature lies in several aspects: $i)$ we analyze energy-efficiency improvement by delaying user equipment's (UE) access to a small cell, and $ii)$ it is shown how UE's  delayed access strategy affect SINR distribution. Access delay is the time between UE's service request and the transmission of data from SBS. Advantage of access delay becomes prominent in small cell networks operating under energy-saving policy since access delay gives UE the opportunity to receive service from a close located cell in an energy-saving mode instead of connecting a distant SBS \cite{HalukRandomOnOff,halukLoadBasedOnOff}.

Contributions of this paper are: $1)$ Expression of coverage probability and transmit rate of a user are derived as a function of transmit range and delay tolerance; $2)$ For the delayed access scheme, it is shown that there exists optimal threshold distance maximizing coverage probability for target SINR; $3)$ An efficient numerical algorithm is developed that optimizes bit-rate by computing optimal threshold distance for given tolerable delay; $4)$ It is shown that considerable improvement in energy-efficiency of small cell network can be achieved by the delayed access scheme. Besides the theoretical work, it is believed that analyzing the access delay vs SINR distribution gives insight into the design of protocols for delay tolerant applications.

The rest of paper arranged as follows: Section II of this paper introduces the system model and definitions of the parameters used throughout the paper. In Section III, we  analyze the distributions of coverage probability  of UE as a function of predefined range and delay. In Section IV, we derive average achievable bit-rate in UE's access scheme. Section V derives optimal transmit range that maximizes the coverage probability, and suggest a numerical algorithm to maximize average bit-rate for given delay budget. In Section VI, we briefly discuss  energy-efficiency improvement by comparing a UE's connection with and without access delay. Section VII provides numerical results that validate our analysis and Section VIII concludes the paper.

\section{Network Model}

In this section, the downlink model of the small cell network is described.
Small base stations are distributed randomly according to homogeneous Poisson point
process (PPP) $\Phi$ of density $\rho_{\rm f}$ in the 2-dimensional plane.
Similarly, distribution of the locations of user equipment (UE) is also PPP
with some intensity.

In our model, SBS has three operation modes, namely, \textit{active},
\textit{idle} and \textit{sleep} modes. At any given time fraction of sleep,
idle and active  cells are $p_{\rm S}$, $p_{\rm I}$, and $p_{\rm A}$, respectively.
In order to preserve ergodicity of the network, we assume these probabilities
do not change. In the active mode, an SBS gives data service to one UE with exponential
service time with mean $1/\mu$. Similarly, duration of the sleep mode is also
exponential with $\lambda_{\rm S}$. Once sleep period expires, sleeping cell is replaced
 with an idle cell. By doing idle-active and idle-sleep exchange, we obtain distribution of active, idle and sleeping cells \cite{IEEEjor:TractableJeff}.
 We assume SBS is \textit{available} if it can serve a service request immediately,
 and \textit{unavailable} if it is in sleep mode, or actively serving another UE.

 After SBS completes a service with all UEs, it moves to the idle mode, and idle cell that is randomly chosen switches to active mode. We assume that there is a central controller that manages such rotation of operation modes among the small cells. In practice, it is clear that the random distribution of sleep cell is not the best choice. Given delay tolerance level, switch-on and switch off decisions can be made by intelligent scheme. Design of such scheme is not the scope of this study; however, it is believed that improvements achieved by delayed access scheme in randomly operated network highlights the minimum possible achievements in energy efficiency.
 
  Wireless channel is modeled by standard propagation model
 with path loss exponent $\alpha > 2$ and Rayleigh fading with unit mean.
 In typical case, the received power at a receiver is $P_{\rm tx} \, g_{i} r^{-\alpha}$,
 where $P_{\rm tx}(i)$ is transmission power of SC $i$, $r_{i}$ is the distance between
 UE and SBS $i$, and $g_{i}$ is the small-scale fading gain following exponential distribution with mean $1/\zeta$.
 Interference power at a receiver is the sum of the received powers from
 all active base stations excluding the SBS that UE is associated with. Then, signal
 to interference and noise ratio is
 \begin{align}
 \text{SINR} = \frac{P_{\rm tx}hr^{-\alpha} }
 {\sum_{i \in \Phi_{\rm A}  \text{\textbackslash\{$b_{\rm o}$\} } }P_{\rm tx}g_{i}r^{-\alpha} + \sigma^2 },
 \label{eq_SINR}
 \end{align}
 where $\sigma^2$ is the noise power of the additive noise, $\Phi_{\rm A}$ denotes HPPP
 for the active cells, and $h$ is the small-scale fading gain associated with the channel between the user of interest and its serving SBS. We assume all active SBSs have constant transmit power,  and interference from idle cells is negligible.

 \section{Analysis of Coverage Probability}

 In this section, we analyze the coverage probability of downlink
 transmission with delayed access in small cell networks. Analysis of coverage probability
 is important because it is highly critical for operators to sustain,
 with high probability,  SINR levels  above a threshold value that
 is acceptable for signal quality. Coverage probability defined as;
 \begin{align}
 p_{\rm c} = \mathbb{P}(\text{SINR}>\gamma ),
 \end{align}
 where $\gamma$ is the threshold SINR. By conditioning on the distance
 between UE and small cell, $p_{\rm c}$ can be written as
 \begin{align}
  p_{\rm c} =
  \int_{0}^{\infty} \mathbb{P}(\text{SINR}>\gamma |r) f_{r,w}(r)\text{d}r,
  \label{eq_pr_cov_generic}
 \end{align}
where $f_{r,w}(r)$ is the probability density function of the distance to nearest idle cell and $w$ is the tolerable delay. Let $t$ be the access time of user to the cell. Clearly we have $t \in [0, w]$.
\subsection{Distance distribution to the nearest available cell}
 Distance to the nearest cell can be evaluated by three independent events: immediate access within threshold distance (IA), delayed access within threshold distance (DA), access outside threshold distance (AO). Let the distance to the nearest base station is $r$. Then, the tail distribution of $r$ in 2-D Poisson process can be written as:
\begin{align}
\mathbb{P}[R < r] = P [ \text{ BS closer than $r$ } ] = 1 - e^{-\pi\rho_{\rm f} r^2},
\label{distCDF_R}
\end{align}
and pdf of $r$ is
\begin{align}
f_{r}(r) = 2\pi \rho_{\rm f} r^2 e^{-\pi\rho_{\rm f} r^2}.
\label{distPDF_R}
\end{align}
In case of IA, number of idle cells Poisson distributed with mean $N_{\rm I} = p_{\rm I}\rho_f \pi R_{\rm th}^2$ for given $R_{\rm th}$. Conditioning on number of idle cells outside an inner circle with radius $r$ within transmission range, tail distribution of distance to the nearest idle cell can be written as
\begin{align}
\mathbb{P}(R > r  | R_{\rm th}, \text{IA})  = & \frac{\sum_{i=1}^{\infty}\left[1 - \frac{r^2}{R_{\rm th}^2} \right]^{i} }{1 - e^{-N_{\rm I}}} \frac{e^{-N_{\rm I}}N_{\rm I}^{i} }{i!}
\label{distCDF_IA}
\\ \nonumber
 = & \frac{1}{1 - e^{-N_{\rm I}} } \left(e^{-N_{\rm I}\frac{r^2}{R_{\rm th}^2} } - e^{-N_{\rm I}}\right).
\end{align}
Distance distribution is independent of the density of small cells in case of DA event. Due to random distribution of sleep and active cells, and the memoriless property of exponential distribution, any cell within $R_{\rm}$ may become available irrespective of  length of waiting time. Then, distance distribution to the cell is
\begin{align}
  \mathbb{P}(R > r |~R_{\rm th}, \text{DA} ) = 1- \frac{r^2}{R_{\rm th}^2}.
\label{distCDF_DA}
\end{align}

In case of OA, CDF of distance to the nearest cell is
\begin{align}
    \mathbb{P}(R > r |R_{\rm th}, \text{OA} ) = \frac{e^{-\rho_{\rm f}\pi (r^2 - R_{\rm th}^2)}}{e^{-N_{\rm I}}}.
\label{distCDF_OA}
\end{align}

After taking the derivative of (\ref{distCDF_IA})-(\ref{distCDF_OA}) with respect to $r$, pdf of distance can be formed as:

\begin{equation}
  f_{R_{\rm th}}(r) = \begin{cases}
     \frac{e^{-p_{\rm I}\rho_{\rm f}\pi r^2 }p_{\rm I}\rho_{\rm f} 2\pi r}{1-e^{p_{\rm I}\rho_{\rm f}\pi R_{\rm th}^2} } & : \text{IA},  (t = 0) \\
    \frac{2r}{R_{\rm th}^2} & : \text{DA},(0<t < w) \\
    \frac{e^{-p_{\rm I}\rho_{\rm f}\pi r^2 }p_{\rm I}\rho_{\rm f} 2\pi r}{e^{p_{\rm I}\rho_{\rm f}\pi R_{\rm th}^2}} &: \text{OA},(t= w)
  \end{cases}~. \label{eq_pdf_distance}
\end{equation}

The probability distribution of distance to nearest cell is not only a function of idle cell density but also the waiting time. Taking into consideration of tolerable delay and by the UE's access time $t$ to the small cell, the respective probabilities of IA, DA, AO events can be written as
\begin{align}
\label{eq_pr_state_transition_IA}
\mathbb{P}\left( \text{IA} \right) & = 1 - e^{-p_{\rm I}\rho_{\rm f}\pi R_{\rm th}^2}  \\
\label{eq_pr_state_transition_DA}
\mathbb{P}\left( \text{DA} \right) & = e^{-p_{\rm I}\rho_{\rm f}\pi R_{\rm th}^2} - e^{-\beta_{\rm w}\rho_{\rm f}\pi R_{\rm th}^2}  \\
\label{eq_pr_state_transition_OA}
\mathbb{P}\left( \text{OA} \right) & = e^{-\beta_{\rm w}\rho_{\rm f}\pi R_{\rm th}^2},
\end{align}
where $\beta_{w}$ is the probability that a cell is either available or will become available within tolerable delay time $w$ (i.e., $\beta_{0} = p_{\rm I}$), which can be easily derived as
\begin{align}
\beta_{w} = 1 - p_{\rm A}e^{-\mu w} - p_{\rm S}e^{-\lambda_{\rm S} w}.
\end{align}

\subsection{Distribution of coverage probability}
Coverage probability can be found by conditioning on distance. Using piece-wise density functions in (\ref{eq_pdf_distance}), and (\ref{eq_pr_state_transition_IA})-(\ref{eq_pr_state_transition_OA}),  we can re-write (\ref{eq_pr_cov_generic}) as
\begin{align}
\mathbb{P}(\text{SINR}&>\gamma )  =  \nonumber\\
&\int_{0}^{R_{\rm th}} \mathbb{P}(\text{SINR}>\gamma |r) f_{R_{\rm th}| \text{IA}}(r,t) \text{d}r 
\nonumber
\times \mathbb{P}\left( \text{IA} \right)\\
+& \int_{0}^{R_{\rm th}} \mathbb{P}(\text{SINR}>\gamma |r) f_{R_{\rm th}| \text{DA}}(r,t) \text{d}r
\times \mathbb{P}\left( \text{DA} \right)
\nonumber \\
+& \int_{R_{\rm th}}^{\infty} \mathbb{P}(\text{SINR}>\gamma |r) f_{R_{\rm th}| \text{OA}}(r,t) \text{d}r
 \times \mathbb{P}\left( \text{OA} \right),
\label{eq_pr_SINR_pwise}
\end{align}

The coverage probability conditioned on the distance as in (\ref{eq_pr_cov_generic}) can then be derived as follows:
\begin{align}
\mathbb{P}(\text{SINR} > \gamma | r)& = \mathbb{P}
 \left\{ \frac{ P_{\rm tx}hr^{-\alpha} }{P_{\rm tx} I_{\rm \Phi_{\rm A}} +\sigma^2 } >\gamma \right\}
 \nonumber\\
 & = \mathbb{E}_{\rm I_{\Phi_{\rm A}}} \left[ \mathbb{P} \left\{
                            h> \frac{\gamma r^\alpha}{P_{\rm tx}}\left( P_{\rm tx}I_{\Phi_{\rm A}}+\sigma^2 \right)  \right\}          \right]
 \nonumber \\
& =  \mathbb{E}_{\rm I_{\Phi_{\rm A}}} \left[ \text{exp}
   \left( -\frac{\zeta \gamma r^\alpha}{P_{\rm tx}}\left( P_{\rm tx}I_{\Phi_{\rm A}}+\sigma^2 \right) \right) \right]
 \nonumber \\
& = e^{-\frac{\zeta \gamma r^\alpha \sigma^2}{P_{\rm tx}}} \mathcal{L}_{\rm I_{\Phi_{\rm A}}}(\zeta \gamma r^\alpha),
\label{eq_pr_cov_dist_cond}
\end{align}
where $\mathcal{L}_{\rm I_{\Phi_{\rm A}}}(s)$ is the Laplace transform of $I_{\Phi_{\rm A}}$ conditioned upon the transmit distance $r$. For $\mathcal{L}_{\rm I_{\Phi_{\rm A}}}(s)$, we have
\begin{align}
\mathcal{L}_{\rm I_{\Phi_{\rm A}}}(s) & =  \mathbb{E}_{\Phi_{\rm A}}\left[ e^{-s\sum_{i \in \Phi_{\rm A}}g_{i}r_{i}^{-\alpha} } \right]
\nonumber \\
& = \mathbb{E}_{\Phi_{\rm A}}\left[  \prod_{i}   e^{-s g_{i}r_{i}^{-\alpha} }   \right]
\nonumber \\
& = \text{exp}(-p_{\rm A}2\pi\rho_{\rm f}\int_{0}^{\infty}\left( 1- \mathbb{E}_{g_{i}}\left[ -s g_{i}z_{i}^{-\alpha} \right]  \right) )z\text{d}z \nonumber \\
& = \text{exp}\left( -p_{\rm A}\pi \rho_{\rm f} r^2 \gamma^{\frac{2}{\alpha}} \int_{0}^{\infty}\frac{\text{d}z}{1+z^{\frac{\alpha}{2}}} \right).
\label{eq_Ltransform_1}
\end{align}
Third equality is due to $\mathbb{E}[ \prod_{x \in \Phi}f(x)]=e^{-\int_{\mathbb{R}^2}(1-f(x)\rho \text{d}x}$  \cite{chiu2013stochastic}. In the last equality of (\ref{eq_Ltransform_1}), we plug $s = \zeta \gamma r^{\alpha}$. As a special case for $\alpha = 4$, one can easily find,
\begin{align}
\mathcal{L}_{\rm I_{\Phi_{\rm A}}}(\zeta \gamma r^4) = e^{-p_{\rm A}\rho_{\rm f}r^2 \sqrt\gamma \pi^2/2}.
\label{eq_Ltransform_SPcase}
\end{align}
We can further derive closed-form solutions with and without white noise.

\subsection{Special cases}
\begin{theorem} For $\alpha=4, \sigma^2 = 0$, coverage probability is
\begin{align}
p_{\rm c} & = \frac{\beta_{0}}{\theta + \beta_{0}}
\nonumber \\
&+ \left( e^{-\beta_{0}v} - e^{-\beta_{\rm w}v} \right)\left[ \frac{1}{\theta v} -e^{-\theta v}\left\{ \frac{\beta_{0}}{\theta + \beta_{0}} + \frac{1}{\theta v}  \right\} \right],
\label{eq_pr_cov_SPalpha4N0}
\end{align}
\end{theorem}

\begin{proof}
\text{see Appendix}
\label{theo_specialCaseWONoise}
\end{proof}


Similarly, interference with noise ($\alpha=4$) can be derived as in \ref{theo_specialCaseWONoise}. Plugging (\ref{eq_Ltransform_SPcase}) in (\ref{eq_pr_cov_dist_cond}), then, substituting  (\ref{eq_pr_state_transition_IA})-(\ref{eq_pr_state_transition_OA}), (\ref{eq_pr_cov_dist_cond}) into (\ref{eq_pr_SINR_pwise}), coverage probability is derived as:
\begin{align}
p_{\rm c}& = \frac{1}{\sigma} \sqrt \frac{\pi P_{\rm tx}}{\zeta \gamma} \Big( p_{\rm I}\pi \rho_{\rm f} a \left[\frac{1}{2} - q(1 - e^{-(\beta_{w} -\beta_{0})v }) \right]
\nonumber \\
&+ \frac{1}{R^2}b (\phi-\frac{1}{2})(e^{-\beta_{0}v} - e^{-\beta_{W}v} )   \Big),
\end{align}
where $a$, $b$, $q$, and $\phi$ are given by:
\begin{align}
a &     = e^{\frac{(p_{\rm A}\sqrt\gamma \pi/2-p_{\rm I})^2(\rho_{\rm f}\pi)^2P_{tx}} {4\zeta\gamma\sigma^2} },\\
b &   = e^{\frac{ (p_A\sqrt \gamma pi^2)^2P_{\rm tx}}{16\zeta\gamma\sigma^2} },\\
q &   = Q \Big(\frac{2\sigma}{\rho_{\rm f}\pi}\sqrt \frac{\zeta\gamma}{P_{\rm tx}}
\big(v+\frac{(p_{\rm A}\gamma \pi/2+ p_{\rm I})(\rho_{\rm f}\pi)^2P_{\rm tx}}{4\zeta \gamma \sigma^2} \big)  \Big),\\
\phi & = \Phi \left( \frac{2\zeta \gamma\sigma^2}{P_{\rm tx} }\left(R^2 +\frac{p_{rm A}\rho_{\rm f} \gamma\pi^2R_{\rm th}^2P_{\rm tx}}{4\zeta\gamma\sigma^2} \right)   \right).
\end{align}

\section{Average Bit Rate for Delayed Access}
In this section, we derive average achievable bit-rate as a function of threshold distance and waiting time by using the respective probabilities of IA, DA and OA (\ref{eq_pr_state_transition_IA})-(\ref{eq_pr_state_transition_OA}). To be more precise, we derive formula in terms of bits/hz expected value of log$_2$(1+SINR) under assumption that Shannon's capacity bound is achievable. Conditioning on access types depending on waiting time of DE, we have
\begin{align}
C  & = \mathbb{E}\left[\text{log$_2$(1 + SINR)} \right] = \mathbb{E}\left[\text{log$_2$(1 + SINR)}\,|\, \text{IA} \right] \mathbb{P}(\text{IA})\nonumber\\
& + \mathbb{E}\left[\text{log$_2$(1 + SINR)} \,|\, \text{DA} \right] \mathbb{P}(\text{DA})\nonumber \\
& + \mathbb{E}\left[\text{log$_2$(1 + SINR)} \,|\, \text{OA} \right] \mathbb{P}(\text{OA}).
\label{eq_bitrate_1}
\end{align}
For each access type (i.e. IA, DA, or OA),conditional expectation needs to be derived. Here, we start with  Average bitrate if UE connects to the cell immediately:
\begin{align}
\mathbb{E}\big[\text{log}_2(1 + & \text{SINR)}\,|\, \text{IA}\big] =
\nonumber \\
& \int_{0}^{R_{\rm th}}  f_{\text{IA}}(r,t)
\int_{0}^{\infty} \frac{\text{ln(1+$\gamma$)}}{\text{ln(2)}} f_{\gamma|r}(\gamma)\text{d}\gamma \text{d}r,
\label{eq_bitrate_IA}
\end{align}
where  $f_{\text{IA}}$ is the density function of distance when DE connects immediately, and $f_{\gamma|r}(\gamma)$ is the conditional density function of SINR given that distance betwen user and transmitting cell is $r$. In the inner integration in (\ref{eq_bitrate_IA}), we can write conditional pdf using the tail distribution of SINR in (\ref{eq_pr_cov_dist_cond}). Then, we have:
\begin{align}
C  & =  \frac{1}{\text{ln(2)}}
\int_{0}^{R_{\rm th}}f_{\text{IA}}(r,t) \int_{0}^{\infty} \text{ln(1+$\gamma$)} \text{d}\left( \mathbb{P}\left(\text{SINR}>\gamma | r \right) \right)\text{d}r
\nonumber\\
& =  \frac{1}{\text{ln(2)}} \int_{0}^{R_{\rm th}} f_{\text{IA}}(r,t)  \int_{0}^{\infty} \frac{\mathbb{P}\left(\text{SINR}>\gamma | r \right)}{1 + \gamma}\text{d}\gamma \text{d}r
\label{eq_bitrate_21}\\
& =  \frac{1}{\text{ln(2)}} \int_{0}^{\infty} \left( \int_{0}^{R_{\rm th}}  \frac{\mathbb{P}\left(\text{SINR}>\gamma | r \right)}{1 + \gamma} f_{\text{IA}}(r,t)\text{d}r\right)  \text{d}\gamma.
\label{eq_bitrate_2}
\end{align} The second equality of (\ref{eq_bitrate_21}) is due to change of variables in the inner integration. After changing the order of integration in third equality, the inner integration becomes equivalent to (\ref{eq_pr_SINR_pwise}). Following similar procedure for $\mathbb{E}\left[\text{log$_2$(1 + SINR )} \,|\, \text{DA} \right]$, and $\mathbb{E}\left[\text{log$_2$(1 + SINR )} \,|\, \text{OA} \right]$, and plugging them in (\ref{eq_bitrate_1}), we can obtain the capacity as:
\begin{align}
C =   \frac{1}{\text{ln(2)}} \int_{0}^{\infty} \frac{\mathbb{P}\left(\text{SINR}>\gamma  \right)} {1 + \gamma}  \text{d}\gamma~.
\end{align}

\section{Optimization of Bit Rate and Coverage Probability for Delayed Access}

In this section, we optimize the coverage probability and average bit rate with respect to threshold distance. Before delving into analytical details, it is helpful to give intuition of why optimal threshold distance exists.  Consider two extreme cases of small and large threshold distances. In the first case, it is not likely that UE finds an idle cell within its threshold distance, causing waste of tolerable delay. On the other case, user is very likely to access an idle cell without much delay, causing minimal use of delay budget. None of these cases yields optimal coverage probability with respect to waiting time. Therefore it is necessary to adjust threshold distance with respect to given delay budget.

\subsection{Optimization of Coverage Probability}
In this section, we optimize coverage probability with respect to given tolerable delay by adjusting threshold distance. For analytical tractability, we made several assumptions. First, we assumed noise power negligible. Second, we assumed cells can adjust their coverage range by increasing or decreasing their transmit power. Notice that coverage probability in (\ref{eq_pr_cov_SPalpha4N0}) is not a function of transmit power. In other words, UE's access strategy does not affect interference power.

For the analysis, we only considered a special case of $\alpha = 4$. To find threshold distance $R_{\rm th}$ that maximizes coverage probability, we first take the derivative of (\ref{eq_pr_cov_SPalpha4N0}) respect to $v$. We then obtain:
\begin{align}
\frac{\text{d} p_{\rm c}}{\text{d} v} & =\beta_{0} e^{-(\beta_{0}+\theta)v} \nonumber\\
&- \frac{1}{\theta v^2} \Big( e^{-\beta_{0}v} - e^{-\beta_{w}v}- e^{-( \beta_{0}+ \theta)v}  + e^{-(\beta_{w} +\theta )v} \Big)
\nonumber \\
& + \frac{1}{\theta v} \Big( -\beta_{0}e^{-\beta_{0}v} +\beta_{w}e^{-\beta_{w}v} + (\beta_{0} + \theta)e^{-(\beta_{0} +\theta)v }\nonumber\\
&-(\beta_{w} +\theta)e^{-(\beta_{w} +\theta)v} \Big)
+\frac{\beta_{0}(\beta_{w}+\theta)}{\beta_{0}+\theta}e^{-(\beta_{w}+\theta)v}.
\label{eq_derivative_pr_cov4}
\end{align}
After substituting second order Taylor expansion for $e^{-\beta_{0}v}$, $e^{-\beta_{w}v}$, $e^{-(\beta_{0}+\theta)v}$, $e^{-(\beta_{w}+\theta)v}$, derivation in (\ref{eq_derivative_pr_cov4}) reduces to quadratic form of $Av^2 +Bv +C$, where $A$, $B$, $C$ are
\begin{align}
A & = \beta_{0}\left( \frac{(\beta_{0}+\theta)^2}{2} + \frac{({\beta_{w} + \theta})^3}{2(\beta_{0}+\theta)} \right)
\nonumber \\
B & = \frac{3}{2} \left(\beta_{0} - \beta_{w})(\beta_{0} + \beta_{w} +\theta) \right) - \beta_{0}\left( \beta_{0} +\theta + \frac{(\beta_{w} +\theta)^2 }{\beta_{0}+\theta} \right)
\nonumber \\
C & =  2(\beta_{w} - \beta_{0} ) + \beta_{0}\left(1 + \frac{\beta_{w}+ \theta}{\beta_{0}+\theta} \right).
\label{eq_quadratik}
\end{align}
Let $v_{1}^{*}$,$v_{2}^{*}$ be the roots of the quadratic quadratic equation in (\ref{eq_quadratik}). Since $\beta_{\rm w} > \beta_{0}$, we have $A>0$, $B<0$, $C>0$, and thus $v_{1}^{*}$,$v_{2}^{*}$  are both positive. Then, candidate threshold distances $R_{\rm th_{1}}^{*}$, $R_{\rm th_{2}}^{*}$ are $\sqrt{\frac{v_{1}^{*}}{\pi \rho_{\rm f}}}$, $\sqrt{\frac{v_{2}^{*}}{\pi \rho_{\rm f}}}$. Finally, among the two threshold distances, we choose the one satisfying second derivative test. To satisfy maximum condition, the second derivative of (\ref{eq_quadratik}) necessitates $v<\frac{B}{2A}$, which is met by the smaller root.

\begin{lemma}
Optimal threshold distance that maximizes the coverage probability is a decreasing function of $\gamma$.

\end{lemma}
\begin{proof}
We have $\gamma \sim \theta^2$. From (\ref{eq_quadratik}), it is clear $A \sim \Theta( \theta^3 ) $, $B \sim \Theta( \theta^2 )$, $C \sim c $, where $c$ is a constant. Then, the roots $\frac{-B\pm \sqrt{B^2-4AC}}{2A}$ are decreasing with respect to $\theta$.
\phantom\qedhere
\end{proof}
\label{lm_roots_decrease}

\subsection{Optimization of Bit-rate}
Optimization of bitrate with respect to threshold distance is not easy to derive analytically. Therefore, we developed an efficient numerical algorithm that computes the optimal threshold distance according to the bitrate.

Algorithm initializes with an upper and a lower bound $R_{\rm u}$, $R_{\rm l}$ respectively. Simply, we choose $R_{\rm l} =0$. Derivation of upper bound can be done based on the threshold distance maximizing the coverage probability. We observe from (\ref{eq_bitrate_2}) that maximizing the coverage probability for given SINR value is strongly related to maximizing bitrate. Then, without loss of generality we can rewrite (\ref{eq_bitrate_2}) in a discrete form as  $\lim_{\Delta \gamma \rightarrow 0} \sum_{i=1}^{\infty} \frac{ \mathbb{P}( \text{SINR}>\gamma_{i} )}{1+\gamma_{i}}\Delta\gamma_{i}$. Let $R_{i}^{*}$ be the optimal threshold distance for $\mathbb{P}( \text{SINR}>\gamma_{i} )$. Then, by choosing optimal thresholds $R_{i}^{*}$,  sum of the bit rates corresponding to each $\gamma_{i}$ is maximized. By lemma \ref{lm_roots_decrease}, we have $R_{i}^{*} > R_{j}^{*}$, where $i< j$. Finally by choosing arbitrarily small $\gamma$ value, an upper bound can be defined. After we define $R_{\rm l}$, $R_{\rm u}$, we can easily find optimal threshold distance by following Algorithm~1.

\begin{algorithm}[!h]
\caption{Algorithm to find optimal $R_{\rm th}$ for bitrate.}
\label{alg_optimizing_bitrate}
\begin{algorithmic}
\State Initialize parameters $\varepsilon$, $R_{\rm l}=0$, $R_{\rm u}$, $\Delta R =R_{u}$
\While {$\Delta R > \varepsilon$}
\State Set $R = (R_{l} + R_{u})/2$ and find search direction by computing bitrates $r_{R + \varepsilon}$, $r_{R - \varepsilon}$
\If {$r_{R + \varepsilon}>$ $r_{R - \varepsilon}$  }
    \State $R_{\rm l} \gets R$
\Else
    \State $R_{\rm u} \gets R$
\EndIf
\State $\Delta R \gets |R- (R_{u}+R_{l})/2|$
\EndWhile
\end{algorithmic}
\end{algorithm}

Note that Algorithm 1 has a precision parameter $\varepsilon$. At each iteration, candidate threshold distance is chosen as the average value of upper and lower bounds of $R_{\rm th}$. Also, search direction is determined by checking the bit-rates in $\varepsilon-$neighborhood of $R$. The algorithm eventually stops  and converges after $\Theta(\log(\frac{R_{u}}{\varepsilon} ))$ steps.

\section{Energy Efficiency Analysis}
In this part, we consider the trade-off between energy-efficiency and delay.
Energy-efficiency is measured as the amount of transmitted bits per unit time
per unit bandwidth per Watt (i.e. bits/s/Hz/Watt). Then energy-efficiency is
\begin{align}
\nu(\rho_{\rm f},\beta_{w}) =\frac{ \tau (\rho_{\rm f},\beta_{\rm w})  }{\text{B $\times$ Mean power consumption of a cell} },\nonumber
\end{align}
where $B$ is the allocated bandwidth. Without loss of generality, we assume
a user uses all available bandwidth during transmission. In order to compare energy efficiency of a cell
with and without delay, we define normalized energy efficiency as
\begin{align}
\nu_{N} = \frac{ \nu(\rho_{\rm f},\beta_{w})}{\nu(\rho_{\rm f},\beta_{0})},
\end{align}
which measures relative scale of improvement in the energy-efficiency with respect to the condition
with no access delay (i.e. $t = 0$).
%

\section{Simulation and Discussions}
In this part, we verify the accuracy of our analysis by Monte-Carlo simulations.
We generate uniformly distributed random variables for X- and Y- coordinates of small cells in a large circle. Cells are independently marked as either idle, active or sleeping with probabilities $p_{\rm I}$, $p_{\rm A}$ and $p_{\rm S}$ respectively. Simulation parameters are listed in Table \ref{table_Sim_params}.
\begin{table}[h]
\centering
\caption{\mbox{Simulation parameters.} }
\begin{IEEEeqnarraybox}[\IEEEeqnarraystrutmode\IEEEeqnarraystrutsizeadd{2pt}{1pt}]{v/c/v/c/v}
\IEEEeqnarrayrulerow\\
&\mbox{ Parameter}&& \mbox{Value} &\\
\IEEEeqnarrayrulerow\\
& \mbox{Small Cell density $ (\rho_{\rm f}$) } && 0.005 \mbox{ per m$^2$} \\
\IEEEeqnarrayrulerow\\
& \mbox{Tolerable delay $ (w_{\rm t})$ } && 10 \mbox{ sec$^{-1}$}\\
\IEEEeqnarrayrulerow\\
& \mbox{Small Cell Transmission Power $(P_{\rm tx})$} && \mbox{23 dBm} \\
\IEEEeqnarrayrulerow\\
& \mbox{Thermal Noise Power $(\sigma^2)$} && \mbox{-104 dBm} \\
\IEEEeqnarrayrulerow\\
& \mbox{Path loss exponent $(\alpha)$} && \mbox{4} \\
\IEEEeqnarrayrulerow\\
& \mbox{Average sleep/active time $(\lambda_{\rm S})$} && \mbox{10 sec}\\
\IEEEeqnarrayrulerow\\
& \mbox{Average sleep/active time $(\mu)$} && \mbox{10 sec}
\\
\IEEEeqnarrayrulerow
\end{IEEEeqnarraybox}
\label{table_Sim_params}
\end{table}

In Figure \ref{fig_pr_cov_thresholdSINR}, coverage probability derived in (\ref{eq_pr_cov_SPalpha4N0}) is verified via simulations. We observe that choosing a proper threshold distance with respect to SINR threshold has significant importance. Especially at low SINR levels, the selection of threshold distance becomes crucial. We observe that coverage probability at 0 dB SINR is almost doubled by choosing proper threshold distance. Especially for delay tolerant data applications, and for intermittent connections, delayed access mechanisms can sustain high coverage probability and improves energy efficiency.
\begin{figure}[h]
\centering
\includegraphics[width=0.5\textwidth]{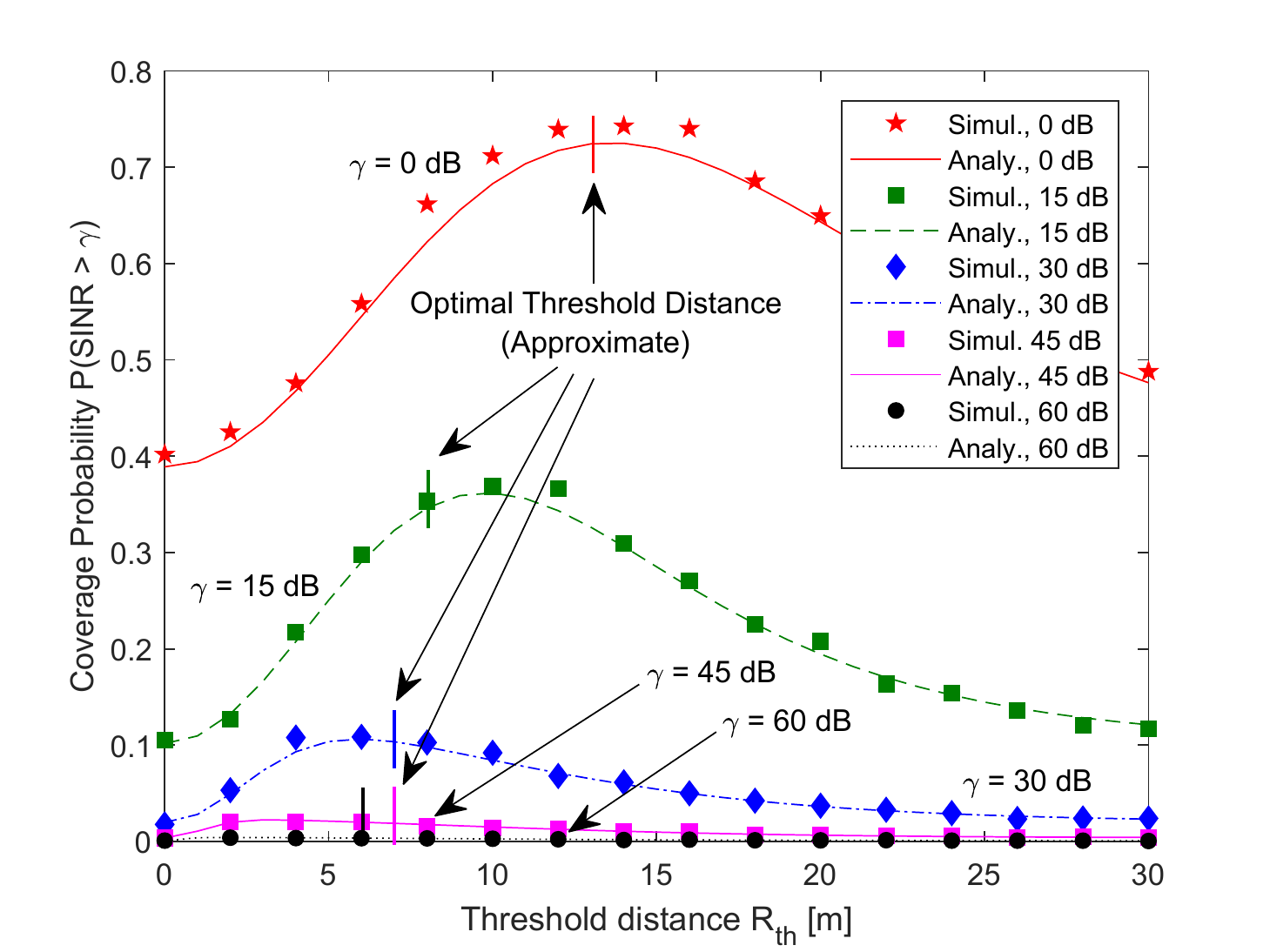}

\caption{Coverage probability at different SINR thresholds with no noise. Threshold distance with respect to given threshold SINR is close to optimum. The small gap between approximation and optimal value is due to Taylor approximation. Tolerable delay $w_{\rm t} = 1/\lambda_{\rm S}$.}
\label{fig_pr_cov_thresholdSINR}
\end{figure}

Figure \ref{fig_bitrate_vs_Rth} shows capacity increase with respect to threshold distance. Tolerable delay is normalized
by the sleeping time. Depending on the fraction of sleep and idle cells, and the improvement in capacity is almost tripled. Capacity gain depends on the transmit power gain. In fact, it is shown in \cite{HalukRandomOnOff} that  transmit power gain is a function of the idle cell density. If the idle cell density is small with respect to active and sleep cell densities, potential transmit power gain is large. In theory, idle cell density can be arbitrarily small compared to active and sleep cell density.

\begin{figure}[h]
\centering
\includegraphics[width=0.5\textwidth]{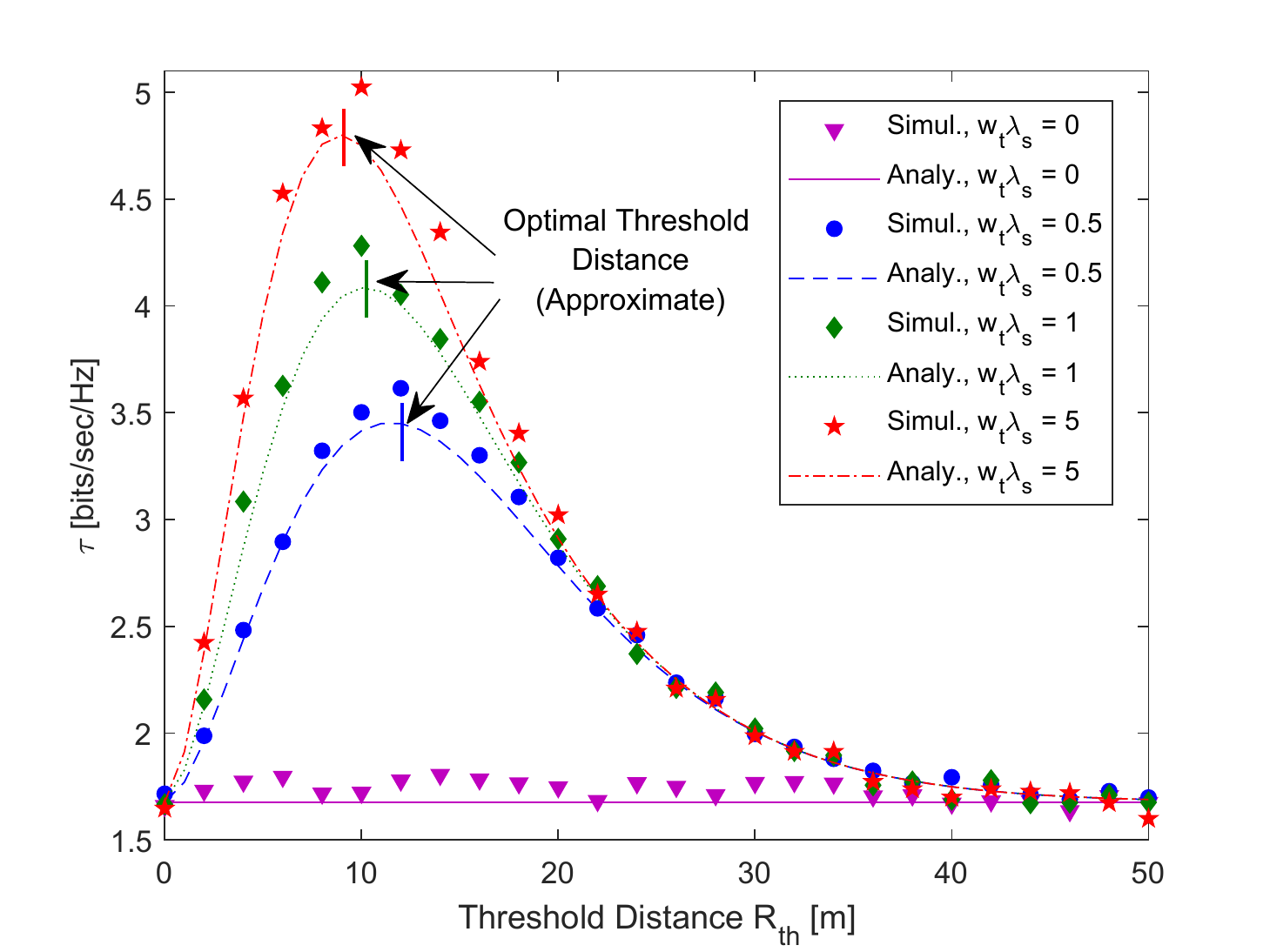}
\caption{Transmission rate for varying threshold distances and delay. Small vertical lines show optimal threshold distance computed by our algorithm.}
\label{fig_bitrate_vs_Rth}
\end{figure}

In Figure \ref{fig_EE_vs_beta_w},  results show the scale of improvement in the network energy efficiency in the proposed scheme with respect to delay starting from zero (i.e. $\beta_0 = 0.1$, delay intolerant case). As the tolerable delay increases proportion of cells that become available, either in time or immediately, $\beta_{w}$, energy-efficiency of network increases.  At high or moderate traffic it is detrimental to turn off cells due to the possibility of degrading quality of service. Because of this, we evaluated network energy efficiency only at low traffic utilization (i.e., $p_{\rm A} =0.1$). For simulation, we kept density of active cells small and fixed and change the density of sleeping and idle cells.  We observe that as the density of sleeping cells with respect to the density of idle cells increase, energy-efficiency of the small cell network increase up to 3-fold.

\begin{figure}[h]
\centering
\includegraphics[width=0.5\textwidth]{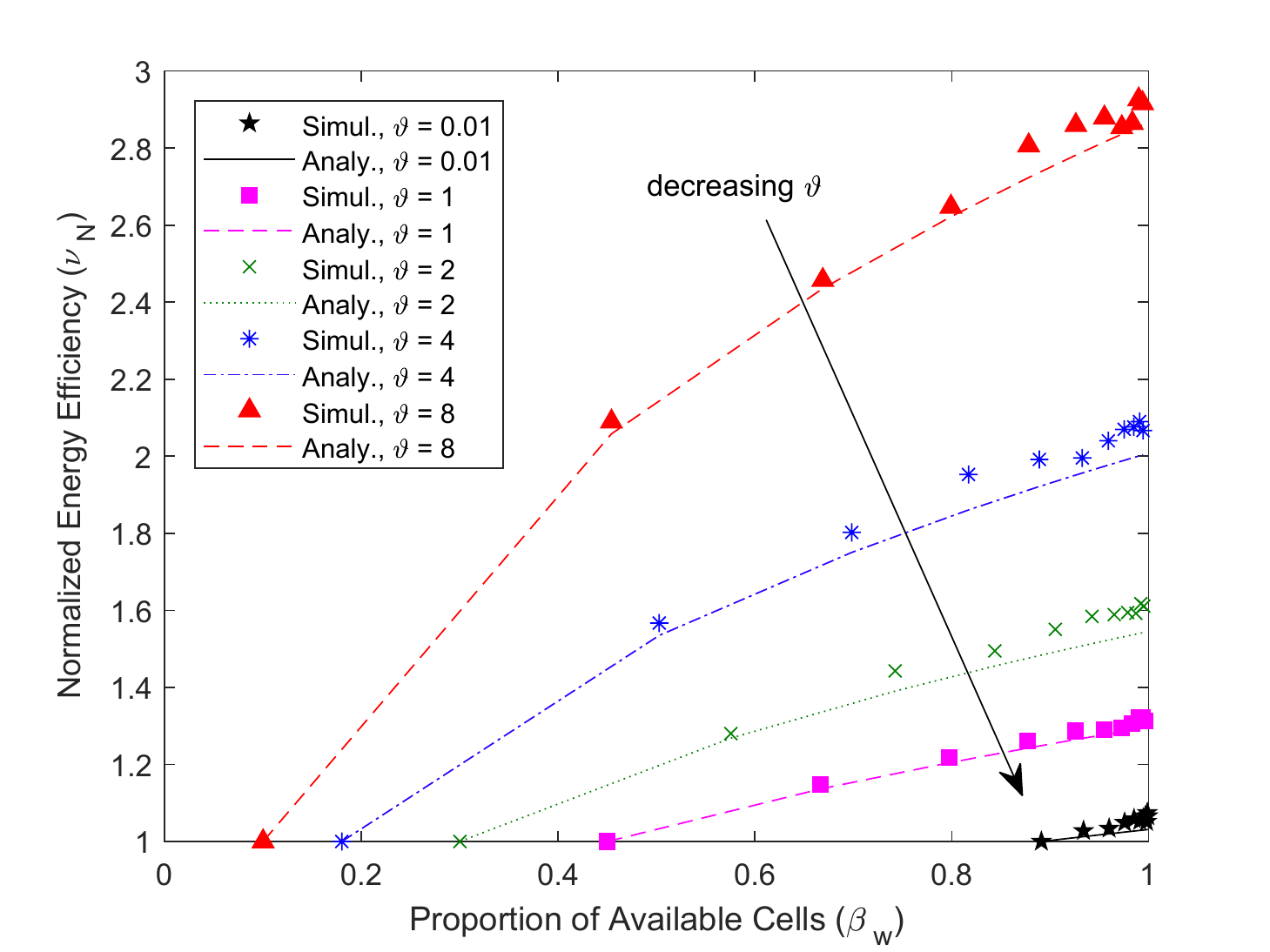}
\caption{Normalized energy efficiency with respect to cell availability along with $\vartheta \,{=}\, p_S/p_I \,{\in}\, \{0.01,1,2,4,8\}$ and $p_{\rm A} \,{=}\, 0.1$.} 
\label{fig_EE_vs_beta_w}
\end{figure}

\section{Conclusion}

In this paper, we propose delayed access mechanism to
improve energy efficiency, coverage probability, and average bit-rate of small cells, and verify its performance
via simulations. For the proposed access mechanism, optimal
threshold distance maximizing the coverage probability is derived.
An efficient numerical algorithm that optimizes threshold distance with respect to the delay budget is developed. Energy-efficiency of the small cell network is assessed by comparing UE's connection to the cell with and without delay. Results show that delayed access strategy can be utilized for  data applications that can tolerate an initial access delay. By delaying UE's access, energy-efficiency of the network can be increased by 3-fold at low traffic utilization.  Some of our future work include designing on/off schemes with
deterministic sleep times, and development of energy efficient
protocols for small cell networks.

\appendix
\begin{proof}
It can be easily seen that (\ref{eq_Ltransform_SPcase}) and (\ref{eq_pr_cov_dist_cond}) becomes equal when $\sigma^2 =0$. Then, by inserting (\ref{eq_pr_state_transition_IA})-(\ref{eq_pr_state_transition_OA}), and (\ref{eq_Ltransform_SPcase}) into (\ref{eq_pr_SINR_pwise}), we have:
\begin{align}
    p_{c} & = \beta_0 \int_{0}^{v}e^{-u(\beta_{0}+\theta)}\text{d}u  \nonumber \\
    &+\frac{1}{v} \int_{0}^{v} e^{-\theta u}\text{d}u\left( e^{-\beta_{0}v} -e^{-\beta_{w}v} \right)\nonumber \\
    &+ \frac{\beta_{0} e^{-\beta_{w}v}}{e^{-\beta_{0}v}}\int_{v}^{\infty}e^{-(\theta +\beta_{0})u}\text{d}u,
\end{align}
where $\theta =p_{\rm A} \sqrt \gamma \pi/2 $, and  $v = \rho_{\rm f}\pi R_{\rm th}^2$. After some algebraic manipulations, coverage probability becomes
\begin{align}
p_{\rm c} & = \frac{\beta_{0}}{\theta + \beta_{0}}
\nonumber \\
&+ \left( e^{-\beta_{0}v} - e^{-\beta_{\rm w}v} \right)\left[ \frac{1}{\theta v} -e^{-\theta v}\left\{ \frac{\beta_{0}}{\theta + \beta_{0}} + \frac{1}{\theta v}  \right\} \right],
\label{eq_pr_cov_SPalpha4N0}
\end{align}
\end{proof}

\bibliographystyle{IEEEtran}
\bibliography{references}

\begin{thebibliography}{10}
\providecommand{\url}[1]{#1}
\csname url@samestyle\endcsname
\providecommand{\newblock}{\relax}
\providecommand{\bibinfo}[2]{#2}
\providecommand{\BIBentrySTDinterwordspacing}{\spaceskip=0pt\relax}
\providecommand{\BIBentryALTinterwordstretchfactor}{4}
\providecommand{\BIBentryALTinterwordspacing}{\spaceskip=\fontdimen2\font plus
\BIBentryALTinterwordstretchfactor\fontdimen3\font minus
  \fontdimen4\font\relax}
\providecommand{\BIBforeignlanguage}[2]{{%
\expandafter\ifx\csname l@#1\endcsname\relax
\typeout{** WARNING: IEEEtran.bst: No hyphenation pattern has been}%
\typeout{** loaded for the language `#1'. Using the pattern for}%
\typeout{** the default language instead.}%
\else
\language=\csname l@#1\endcsname
\fi
#2}}
\providecommand{\BIBdecl}{\relax}
\BIBdecl

\bibitem{5Gwhitepaper}
E.~NetWorld2020, ``5{G}: Challenges, research priorities, and
  recommendations,'' \emph{Joint White Paper September}, 2014.

\bibitem{Magazine:5GChallengesSON}
A.~Imran, A.~Zoha, and A.~Abu-Dayya, ``Challenges in {5G}: how to empower son
  with big data for enabling {5G},'' \emph{IEEE Network}, vol.~28, no.~6, pp.
  27--33, 2014.

\bibitem{IEEEjor:TractableJeff}
J.~G. Andrews, F.~Baccelli, and R.~K. Ganti, ``A tractable approach to coverage
  and rate in cellular networks,'' \emph{IEEE Trans. Commun.}, vol.~59, no.~11,
  pp. 3122--3134, 2011.

\bibitem{StochGeoOnOff}
Z.~Lin, Y.~Gao, B.~Gong, X.~Zhang, and D.~Yang, ``Stochastic geometry study on
  small cell on/off adaptation,'' in \emph{Proc. Int. Conf. Commun. Netw. in
  China}, Aug 2014, pp. 331--334.

\bibitem{smallCellSleepStrategies}
C.~Liu, B.~Natarajan, and H.~Xia, ``Small cell base station sleep strategies
  for energy efficiency,'' \emph{IEEE Trans. Vehic. Technol.}, vol.~65, no.~3,
  pp. 1652--1661, March 2016.

\bibitem{URLLCBennisLetter}
C.~Liu and M.~Bennis, ``Ultra-reliable and low-latency vehicular transmission:
  An extreme value theory approach,'' \emph{IEEE Communications Letters},
  vol.~22, no.~6, pp. 1292--1295, June 2018.

\bibitem{DBLP_FedLearnURLLV2V}
\BIBentryALTinterwordspacing
S.~Samarakoon, M.~Bennis, W.~Saad, and M.~Debbah, ``Federated learning for
  ultra-reliable low-latency {V2V} communications,'' \emph{CoRR}, vol.
  abs/1805.09253, 2018. [Online]. Available:
  \url{http://arxiv.org/abs/1805.09253}
\BIBentrySTDinterwordspacing

\bibitem{HalukRandomOnOff}
H.~\c{C}elebi, N.~Maxemchuk, Y.~Li, and I.~G. ven\c{c}, ``Energy reduction in
  small cell networks by a random on/off strategy,'' in \emph{Proc. IEEE
  Globecom Workshops (GC Wkshps)}, Dec. 2013, pp. 176--181.

\bibitem{halukLoadBasedOnOff}
H.~\c{C}elebi and I.~Guven\c{c}, ``Load analysis and sleep mode optimization
  for energy-efficient {5G} small cell networks,'' in \emph{Proc. IEEE Int.
  Conf. Commun. Workshops (ICC Workshops)}, May 2017, pp. 1159--1164.

\bibitem{chiu2013stochastic}
S.~N. Chiu, D.~Stoyan, W.~S. Kendall, and J.~Mecke, \emph{Stochastic geometry
  and its applications}.\hskip 1em plus 0.5em minus 0.4em\relax John Wiley \&
  Sons, 2013.

\end{thebibliography}
\end{document}